\documentclass[10pt,a4paper]{article}
\usepackage[margin=0.88in]{geometry}
\usepackage[T1]{fontenc}
\usepackage{lmodern}
\usepackage{amsmath,amssymb,amsthm,booktabs,array,microtype}
\usepackage{tikz}
\usepackage{hyperref}
\hypersetup{colorlinks=true,linkcolor=black,citecolor=blue,urlcolor=blue,pdftitle={From Link Availability to Deadline-Constrained Key Service},pdfauthor={Hasan Abbas Al-Mohammed}}
\usepackage{enumitem}
\setlist[itemize]{leftmargin=*,topsep=2pt,itemsep=1pt}
\newtheorem{proposition}{Proposition}
\newtheorem{observation}{Observation}
\newcommand{\pp}[1]{\left[#1\right]_{+}}
\title{\bfseries From Link Availability to Deadline-Constrained Key Service:\\A Feasibility Framework for Mobile Free-Space Quantum Key Distribution}
\author{Hasan Al-Mohammed}

\begin{document}
\maketitle
\begin{abstract}
Mobile free-space quantum key distribution (QKD) links can provide secret bits intermittently, whereas applications may need cryptographic key material at specified times. Optical feasibility, positive finite-key output, and deadline-compliant key service are consequently different claims. This technical note formalizes their separation through a time-indexed, security-qualified arrival process and a key-inventory model. Under explicitly stated assumptions, an elementary prefix-feasibility condition characterizes whether all demands can be served without borrowing keys from future contact windows. A stylized finite-block model then illustrates why positive long-run optical availability does not ensure positive finite-key output when observations are fragmented into short contacts. Two transparent numerical examples distinguish total key yield from key delivery time and show the conditional benefit of pooling data across acquisition windows. The analysis is intended as a systems-level accounting and benchmarking framework, not a new QKD security proof or a demonstration of hardware performance. It identifies measurements and security assumptions needed to evaluate train, UAV, high-altitude-platform, and satellite-assisted optical QKD against application-level service requirements.
\end{abstract}

\section{Introduction}
Security protocols and operational services are evaluated with different metrics. The BB84 protocol establishes a basis for quantum key distribution (QKD) \cite{bb84}, decoy-state methods support weak-coherent implementations \cite{lo2005}, and finite-key security analyses quantify the price of finite observations and post-processing \cite{lim2014,scaranirenner}. A network user, however, does not directly consume received photons, successful detection events, or an asymptotic secret-key rate. It consumes available, authenticated, usable key material at a particular time.

This discrepancy is especially important for free-space optical (FSO) links, whose collection efficiency, extinction, pointing, and visibility are geometry and time dependent \cite{khalighi2014,henniger2010}. Train studies have examined contact geometry, weather, and the combination of FSO with QKD \cite{train2022,train2023,train2024,uhst2022}; UAV studies illustrate power-weather tradeoffs \cite{uav2024}; satellite experiments show non-terrestrial QKD feasibility under specific conditions \cite{liao2017}. Relay-based high-altitude platform (HAP) architectures add trust boundaries and per-hop availability \cite{hap2025,hap2026}. None of these application settings licenses the inference that an optical link with positive received power always supplies a ready-to-use secret key.

Early protocol simulation work and an ML-oriented study of key-length effects further motivate the distinction between computed or predicted quantities and a formal security certificate \cite{simqkd2021,keylength2024}. A recent interface study links classical FSO budgets to decoy-state finite-key estimates \cite{bridge2026}. The present note addresses a \emph{different, subsequent question}: even if that calculation produces nonnegative, security-qualified keys during contact windows, are those keys available when applications need them? The distinction is also relevant to threat monitoring and adaptive processing: classifiers and learned estimators can support operation but are not substitutes for the physical assumptions and proofs used to certify a secret key \cite{detect2021,nn2021,cascade2024,defense2026,specialized2026}.

\paragraph{Contributions and boundaries.} We (i) distinguish optical-link feasibility, block-level security-qualified output, and deadline-constrained service; (ii) state and prove a simple prefix-demand feasibility criterion under idealized inventory assumptions; (iii) analyze finite-block fragmentation with an illustrative surrogate, without interpreting it as a cryptographic bound; and (iv) specify a reproducible data ledger for operational comparisons. The mathematical statements are elementary systems accounting, not new security theorems. We introduce no new QKD protocol, hardware experiment, or validated rate prediction.

\section{Three feasibility gates and the scope of prior work}
For a moving or intermittent optical link, \textit{three distinct gates} should be evaluated in order (Fig.~\ref{fig:gates}).
\begin{enumerate}[leftmargin=*,itemsep=1pt]
\item \textbf{Physical contact:} the source, receiver, geometry, pointing, atmosphere and acquisition process yield a valid observation window.
\item \textbf{Secret-key certification:} the chosen QKD protocol, authenticated classical communication, device model and finite-key analysis yield a nonnegative composably secure output \emph{under their stated assumptions}.
\item \textbf{Service delivery:} the resulting keys are present in the right trust domain before demand deadlines, after any reservations or usage overheads.
\end{enumerate}
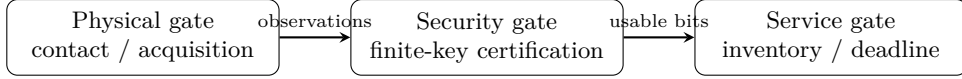
\begin{figure}[t]
\centering
\begin{tikzpicture}[>=stealth, every node/.style={font=\small}]
\node[draw,rounded corners,align=center,minimum width=3.6cm,minimum height=1.0cm] (p) at (0,0) {Physical gate\\contact / acquisition};
\node[draw,rounded corners,align=center,minimum width=3.6cm,minimum height=1.0cm] (s) at (4.55,0) {Security gate\\finite-key certification};
\node[draw,rounded corners,align=center,minimum width=3.6cm,minimum height=1.0cm] (d) at (9.1,0) {Service gate\\inventory / deadline};
\draw[->,thick] (p) -- node[above,font=\scriptsize] {observations} (s);
\draw[->,thick] (s) -- node[above,font=\scriptsize] {usable bits} (d);
\end{tikzpicture}
\caption{A positive result at one gate does not imply a positive result at the next. The boxes describe an analysis sequence, not a new QKD security proof.}\label{fig:gates}
\end{figure}

\paragraph{Related classical-security and quantum-technology boundaries.}
A practical key service also relies on a classical control plane, but its performance and trust requirements must be assessed with task-specific evidence. A prior comparison of CPU and GPU hardware for cryptocurrency mining \cite{mining2020} investigates a different computational workload; it is relevant here only as an example of why the throughput, energy, and resource costs of computing platforms must be measured for the \emph{actual} workload. Cryptocurrency-mining benchmarks and financial comparisons cannot be used to estimate QKD error-correction or privacy-amplification time. Likewise, blockchain-based vehicle authentication and traffic-event validation \cite{vehicleblockchain2020} address identity, provenance, and trust in transportation systems, not the composable secrecy of QKD-generated key bits. Neither application replaces the authenticated classical communication and security analysis required in QKD.

A further scope distinction concerns the physical task itself. Survey-level discussion of quantum-computer architecture and cryptanalytic applications \cite{qcomputer2020} gives broader motivation for quantum-safe communications, whereas quantum radar \cite{qradar2020} concerns sensing rather than secret-key distribution. These related quantum technologies are \emph{not} evidence of optical-contact availability, finite-key security, or deadline service in the present framework. We include them to make the system and security boundaries explicit, not to extrapolate their findings into QKD claims.

Table~\ref{tab:prior} maps neighboring technical questions to representative research. The listed studies make different claims; no common numerical benchmark or experimental equivalence is implied.
\begin{table}[t]
\centering\small
\caption{Selected existing evidence and the additional time-domain question.}\label{tab:prior}
\begin{tabular}{p{0.215\linewidth}p{0.365\linewidth}p{0.335\linewidth}}
\toprule
\textbf{Question} & \textbf{Representative literature} & \textbf{Still needed for service claims} \\
\midrule
Contact feasibility & FSO propagation surveys \cite{khalighi2014,henniger2010}; mobile train and UAV models \cite{train2022,train2023,uav2024} & Acquired-window traces, realistic outages, time-varying pointing and visibility.\\
QKD conversion & Decoy states and finite-key methods \cite{lo2005,lim2014}; railway QKD \cite{train2024}; link-budget interface \cite{bridge2026} & Secure bits \emph{per acquired window}, with protocol- and trust-qualified assumptions.\\
Relays and routes & MDI-QKD \cite{mdi2012}; satellite QKD \cite{liao2017}; trusted-HAP studies \cite{hap2025,hap2026} & Per-hop key ownership, trusted-node boundaries, routing and replenishment delays.\\
Adaptive operations & Detector learning \cite{detect2021,nn2021}, Cascade-oriented ML \cite{cascade2024}, synthesis \cite{defense2026,specialized2026} & Explicit separation of prediction/selection from proof-derived key bounds.\\
Application context & QKD for IoT/6G \cite{iot2021,master2021,applications2024}; network convergence \cite{metaverse2025} & Time-indexed key demand, buffering, authentication cost and deadlines.\\
\bottomrule
\end{tabular}
\end{table}

\section{A time-indexed key-service model}
Partition the operating horizon into slots $t=1,\dots,T$, chosen so that key-delivery and application-demand timing can be represented. For each acquired contact window $j$, let $\mathcal{D}_j$ denote the measured and authenticated data entering an \emph{applicable} QKD security analysis. Its certified usable yield is
\begin{equation}
K_j = \pp{\ell_{\rm sec}(\mathcal{D}_j;\epsilon_j)-c_j},\qquad K_j\geq 0,
\label{eq:yield}
\end{equation}
where $\ell_{\rm sec}$ is the output of a specified finite-key treatment and $c_j\geq0$ denotes any key bits separately reserved for authentication or other required uses \emph{if not already included} in that treatment. The notation deliberately leaves the security proof external; it must not be evaluated by simply inserting classical optical throughput into $\ell_{\rm sec}$. If a valid analysis produces no secret key, set $K_j=0$. An invalid model is a failure of the security gate, not a zero-bit security certificate.

Each yield arrives only once completed post-processing makes it usable. Define $A_t=\sum_{j:\,\tau_j=t}K_j$, where $\tau_j$ is the first slot in which the key can be consumed. Let $D_t\geq0$ be the application's key demand due in slot $t$ and $B_0\geq0$ be usable initial inventory \emph{in the same trust domain}. We assume arrivals credited at the beginning of slot $t$ are available before $D_t$ must be served. Under full service, the inventory recurrence is
\begin{equation}
B_t=B_{t-1}+A_t-D_t,\qquad B_t\geq0.
\label{eq:inventory}
\end{equation}
The model assumes unlimited storage, no key expiration, no key transfers across mutually untrusted domains, and irreversible consumption of key material. If a task needs $q$ fresh key bits, it must receive $q$ appropriately generated bits; key reuse is not implicit. Finite classical-control latency can be represented in $\tau_j$.

\begin{proposition}[Prefix-feasibility criterion]\label{prop:prefix}
Under the assumptions above, all key demands $D_t$ can be met by their due slots if and only if
\begin{equation}
\boxed{\quad B_0+\sum_{t=1}^{u}A_t \;\geq\; \sum_{t=1}^{u}D_t\quad \text{for every }u=1,\ldots,T.\quad}
\label{eq:prefix}
\end{equation}
\end{proposition}
\begin{proof}
Summing the inventory equation through any prefix $u$ yields $B_u=B_0+\sum_{t=1}^{u}(A_t-D_t)$. Nonnegative inventory at every due slot therefore requires Eq.~\eqref{eq:prefix}. Conversely, if Eq.~\eqref{eq:prefix} holds for every $u$, the same expression gives $B_u\geq0$ for every slot, so the inventory can meet each demand without using future arrivals.
\end{proof}

The result is a deterministic bookkeeping test rather than a novel cryptographic theorem. It shows why $\sum_t A_t\geq\sum_t D_t$ by itself is inadequate: a secret key generated after a deadline cannot retroactively serve that deadline. With hard storage limits, expiration, network routing, or differentiated key pools, Eq.~\eqref{eq:prefix} remains a necessary condition for the aggregate pool but is generally not sufficient for the constrained system.

\paragraph{Robust demand and uncertain contact yields.} Suppose prospective security-qualified arrivals and demands obey independently selectable bounds $\underline A_t\le A_t\le\overline A_t$ and $\underline D_t\le D_t\le\overline D_t$. To guarantee all deadlines for \emph{every} sequence within this rectangular uncertainty set, Eq.~\eqref{eq:prefix} is equivalent to
\begin{equation}
 B_0+\sum_{t=1}^{u}\underline A_t\;\geq\;\sum_{t=1}^{u}\overline D_t
 \qquad (u=1,\ldots,T).
 \label{eq:robust}
\end{equation}
Necessity follows by selecting minimum arrivals and maximum demands up to a violating prefix; sufficiency follows because all other admissible prefixes only increase supply or decrease demand. The independence/rectangular assumption matters: correlated weather states, coupled links and constrained demand processes can make this a conservative test. The lower arrival values must be \emph{valid service-deliverable lower bounds}, not optimistic means or classical received-power predictions. A bound of zero for a fully blocked contact contributes no secret keys.

\paragraph{A useful deficit diagnostic.} For a planned schedule define
\begin{equation}
\Gamma=\max_{1\le u\le T}\pp{\sum_{t=1}^{u}D_t-B_0-\sum_{t=1}^{u}A_t}.
\label{eq:deficit}
\end{equation}
$\Gamma$ is the smallest \emph{additional initial inventory} required to satisfy all deadlines in this idealized single-pool model. This follows directly by increasing $B_0$ in Eq.~\eqref{eq:prefix}. It is not a measure of cryptographic secrecy or proof failure.

\section{Finite-block fragmentation: an illustrative threshold}
A channel may be physically available for several brief contacts without yielding positive per-contact secret keys. In a particular decoy-state protocol, the relevant penalty depends on observed counts, parameter-estimation methods, leakage, security parameters, and proof-specific finite-size terms \cite{lim2014,scaranirenner}. To illustrate the \emph{shape} of the tradeoff without claiming a validated QKD bound, introduce the surrogate
\begin{equation}
\widetilde{\ell}(n)=a n-b\sqrt{n}-c,
\qquad a>0,\quad b,c\geq0,
\label{eq:surrogate}
\end{equation}
where $n$ is an abstract processed-block size, $an$ an asymptotic contribution, and the remaining terms stand for schematic finite-size and fixed penalties. No particular QKD protocol has been shown to obey Eq.~\eqref{eq:surrogate} with the illustrative coefficients below.

\begin{observation}[Positive-output threshold for the surrogate]
The surrogate in Eq.~\eqref{eq:surrogate} is positive precisely when
\begin{equation}
n>n_{\rm crit}:=\left(\frac{b+\sqrt{b^2+4ac}}{2a}\right)^2.
\label{eq:threshold}
\end{equation}
\end{observation}
\noindent Indeed, setting $x=\sqrt{n}\geq0$ reduces the condition to $ax^2-bx-c>0$, which holds above the positive root. This threshold is \emph{not} a protocol-independent minimum number of QKD pulses; it is a property of the stated surrogate only.

For $m$ equal blocks of size $n$, treating every block separately gives $m\max\{0,\widetilde\ell(n)\}$, whereas treating the same $mn$ observations as one block gives $\max\{0,\widetilde\ell(mn)\}$. The latter may be higher \emph{if} aggregation is permitted by the actual QKD protocol, proof, stationarity assumptions, and parameter-estimation procedure. Combining samples across distinct channel states, devices, basis choices or trust domains without an applicable proof would not be valid. In particular, this note does not propose aggregating certified outputs across incompatible QKD proofs.

\section{Transparent illustrative cases}
\subsection{Same total key supply; different delivery time}
Consider six abstract slots with $D_t=2$ kbit per slot and $B_0=6$ kbit. In Scenario A a valid contact produces $10$ kbit of usable, certified key at the beginning of slot 3; in Scenario B the same quantity arrives only at slot 6. In both cases total supply, including initial inventory, is $16$ kbit against a horizon demand of $12$ kbit.
\begin{table}[h]
\centering\small
\caption{Illustrative key-service ledger (kbit). These are assumed inputs, not simulated or measured QKD rates.}\label{tab:schedule}
\begin{tabular}{c*{6}{r}r}
\toprule
& $t=1$ & $2$ & $3$ & $4$ & $5$ & $6$ & Total\\
\midrule
Demand $D_t$ & 2 & 2 & 2 & 2 & 2 & 2 & 12\\
Scenario A arrival $A_t$ & 0 & 0 & 10 & 0 & 0 & 0 & 10\\
Scenario B arrival $A_t$ & 0 & 0 & 0 & 0 & 0 & 10 & 10\\
\bottomrule
\end{tabular}
\end{table}
Scenario A passes every prefix test: $B_0=6$ kbit covers the first two demands before replenishment. Scenario B first fails at slot 4 because cumulative demand reaches $8$ kbit while available cumulative supply is only $6$ kbit. Its worst deficit is $4$ kbit at slot 5, hence Eq.~\eqref{eq:deficit} gives $\Gamma_B=4$ kbit. A report consisting only of total secret bits across the horizon would miss this service outage.

\subsection{Fragments too short to clear a schematic finite-block threshold}
For a second, entirely illustrative case let $a=0.025$, $b=4$, $c=100$ in Eq.~\eqref{eq:surrogate}. Four contacts each provide $n=10{,}000$ abstract data units. Processing one contact at a time yields $\widetilde\ell(10{,}000)=250-400-100=-250$ bits, and thus no positive output from any block. If the data can lawfully form a single $40{,}000$-unit block, the surrogate gives $\widetilde\ell(40{,}000)=1{,}000-800-100=100$ bits. This contrast is \emph{only a numerical illustration of a finite-block effect}; neither set of values constitutes a finite-key security statement. Correct evaluation requires the real, protocol-specific counting statistics and security parameters.

\section{Toward an auditable benchmark for mobile FSO/QKD}
The framework is intentionally small enough to support a reproducible comparison without requiring a new security claim. For each time interval or contact window, a study should report at least the following:
\begin{enumerate}[leftmargin=*,itemsep=2pt]
\item \textbf{Contact trace:} source-receiver geometry, weather/attenuation or measured transmission, acquisition start and end times, pointing loss and outages. Models and field measurements must be identified separately.
\item \textbf{QKD observations:} source and decoy settings, emitted pulses, detection counts by relevant basis/intensity, QBER or equivalent observables, detector characteristics and accepted-window criteria.
\item \textbf{Security accounting:} specified finite-key proof, security failure budget, leakage, authenticated classical communication and explicit treatment of any blocks that fail to produce positive keys.
\item \textbf{Service accounting:} when each key becomes usable; trust domain, relay assumptions, initial inventory, application demand timestamps, latency, reserved authentication keys and discarded or expired bits. Report measured classical-processing duration and energy for the particular reconciliation, authentication, and privacy-amplification workload; do not substitute benchmarks from unrelated computational applications \cite{mining2020}.
\item \textbf{Outcomes:} secure bits per contact, percentage of contacts with $K_j>0$, total available keys, worst prefix deficit $\Gamma$, fraction of demand met by deadlines and any key-buffer constraints.
\end{enumerate}
This separation supports controlled sensitivity analyses: varying aperture or weather should change the physical input, not silently alter the key security definition; changing a machine-learning detector should change its operational decision metric, not imply a new secrecy proof \cite{detect2021,nn2021,defense2026}. A trusted HAP relay and an untrusted MDI station require different security models and cannot be assigned the same key ownership ledger \cite{mdi2012,hap2025,hap2026}. The PLOB rate-loss bound provides an additional context for assessing claims made for repeaterless quantum channels \cite{plob2017}.

\section{Limitations, non-overlap and future validation}
This note is deliberately narrower than a full systems experiment. There is no stochastic-optimal policy, calibrated outage model, or observed data trace. The queueing criterion assumes a single compatible inventory, unlimited storage and immediate crediting at the stated arrival slot. The finite-block expression in Eq.~\eqref{eq:surrogate} is not a security bound. Any operational result must use a specified finite-key treatment and verified input assumptions; an AI estimator cannot replace certified parameter estimation \cite{lim2014,defense2026,specialized2026}.

The manuscript also has a scope boundary relative to earlier work. The FSO-to-decoy-state study \cite{bridge2026} addresses the physical-to-key conversion; the present analysis begins with the \emph{output} of such a conversion and tests service timeliness. Prior HAP work \cite{hap2026} examines conditional relay availability and finite-key yields; it does not establish the single-pool deadline-feasibility result for application demand across time slots. Claims of empirical advantage, network-wide optimality or experimentally certified rate are therefore outside this manuscript.

An evidence-based next step is to combine measured acquisition windows with actual QKD event logs and application demand traces; compare fixed and adaptive scheduling under the same cryptographic proof; and test separate trust-domain buffers, authentication costs and hard key expirations. The model should be extended for multi-hop routing and delayed post-processing before claiming real-world service guarantees.

\section{Conclusion}
Practical QKD service requires more than an optical link that sometimes connects and more than a positive total secret-key yield. It requires correctly modeled security-qualified arrivals before application deadlines. A simple prefix-feasibility test makes that distinction auditable, while a stylized finite-block example demonstrates why contact fragmentation can matter even when total observation time is unchanged. These accounting tools neither strengthen nor replace a cryptographic proof. Their intended role is to connect physical-link studies and security calculations to explicitly time-resolved service evaluation.

\section*{Disclosure and manuscript status}
This is a new analytical draft for author verification. The illustrative data are synthetic and are not experimental results. The author should independently verify the technical argument, every reference, the novelty relative to prior work, arXiv category suitability and any applicable disclosure policies before submission.

\end{document}